\documentclass[acmsmall,nonacm]{acmart}

\usepackage[capitalize,noabbrev,nameinlink]{cleveref}
\usepackage{listings}
\usepackage{microtype}

\AtBeginDocument{%
  }

\setcopyright{acmcopyright}
\copyrightyear{2023}
\acmYear{2023}
\acmDOI{XXXXXXX.XXXXXXX}

\acmConference[PODC 2023]{ACM Symposium on Principles of Distributed Computing}{June 19--23, 2023}{Orlando, FL}
\acmPrice{15.00}
\acmISBN{978-1-4503-XXXX-X/18/06}

\newcommand{\mypar}[1]{\noindent\textbf{#1}}
\newcommand{\myskip}{\vspace{2pt}}

\definecolor{Azure1}{rgb}{0.95,1,1}

\lstdefinestyle{mystyle}{
    backgroundcolor=\color{Azure1},
    columns=fixed,
    basicstyle=\ttfamily\small,
    basewidth=0.5em,
    breakatwhitespace=false,
    breaklines=true,
    captionpos=b,
    frame=single,
    keepspaces=true,
    numbers=left,
    numberstyle=\small,
    numbersep=5pt,
    showspaces=false,
    showstringspaces=true,
    showtabs=false,
    tabsize=2
}

\lstdefinestyle{mystyle2}{
    backgroundcolor=\color{Azure1},
    columns=fixed,
    basicstyle=\ttfamily\tiny,
    basewidth=0.55em,
    breakatwhitespace=false,
    breaklines=true,
    captionpos=b,
    frame=single,
    keepspaces=true,
    numbers=left,
    numberstyle=\tiny,
    numbersep=5pt,
    showspaces=false,
    showstringspaces=true,
    showtabs=false,
    tabsize=2
}

\crefname{lstlisting}{Listing}{Listings}

\newcommand{\al}[2]{\cref{l:#1}, Line~#2}
\newcommand{\als}[2]{\cref{l:#1}, Lines~#2}

\begin{document}

\title{Efficient Computation in Congested Anonymous Dynamic Networks}

\author{Giuseppe A. Di Luna}
\authornote{Both authors contributed equally to this research.}
\email{g.a.diluna@gmail.com}
\affiliation{%
  \institution{DIAG, Sapienza University of Rome}
  \city{Rome}
  \country{Italy}
}

\author{Giovanni Viglietta}
\authornotemark[1]
\email{viglietta@gmail.com}
\orcid{0000-0001-6145-4602}
\affiliation{%
  \institution{Department of Computer Science and Engineering, University of Aizu}
  \city{Aizuwakamatsu}
  \country{Japan}
}

\renewcommand{\shortauthors}{G.\,A. Di Luna and G. Viglietta}

\begin{abstract}
An \emph{anonymous dynamic network} is a network of indistinguishable processes whose communication links may appear or disappear unpredictably over time. Previous research has shown that deterministically computing an arbitrary function of a multiset of input values given to these processes takes only a linear number of communication rounds (Di Luna--Viglietta, FOCS~2022).

However, fast algorithms for anonymous dynamic networks rely on the construction and transmission of large data structures called \emph{history trees}, whose size is polynomial in the number of processes. This approach is unfeasible if the network is \emph{congested}, and only messages of logarithmic size can be sent through its links. Observe that sending a large message piece by piece over several rounds is not in itself a solution, due to the anonymity of the processes combined with the dynamic nature of the network. Moreover, it is known that certain basic tasks such as all-to-all token dissemination (by means of single-token forwarding) require $\Omega(n^2/\log n)$ rounds in congested networks (Dutta et al., SODA~2013).

In this work, we develop a series of practical and efficient techniques that make it possible to use history trees in congested anonymous dynamic networks. Among other applications, we show how to compute arbitrary functions in such networks in $O(n^3)$ communication rounds, greatly improving upon previous state-of-the-art algorithms for congested networks.
\end{abstract}

\ccsdesc[500]{Theory of computation~Distributed algorithms}
\ccsdesc[500]{Computing methodologies~Distributed algorithms}

\keywords{anonymous dynamic network, congested network, history tree}


\maketitle

\section{Introduction}\label{s:1}

\mypar{Dynamic networks.} In recent years, distributed computing has seen a remarkable increase in research on the algorithmic aspects of networks that constantly change their topology~\cite{CFQS12,KO,MS18}. The study of these \emph{dynamic networks} is motivated by technologies such as wireless sensors networks, software-defined networks, and networks of smart devices. Typically, the distributed system consists of $n$ processes that communicate with each other in synchronous rounds. At each round, the network topology is rearranged arbitrarily, and communication links appear or disappear unpredictably.

\myskip
\mypar{Anonymity and leadership.} There are efficient algorithms for various tasks that work under the assumption that processes have \emph{unique IDs}~\cite{CFMS15,JY24,KLO11,KLO10,KOM11,MS18,DW05}. However, unique IDs may not be available due to operational limitations~\cite{DW05} or to protect user privacy; for instance, assigning temporary random IDs to users of COVID-19 tracking apps was not sufficient to eliminate privacy concerns~\cite{SM20}. Systems where processes are indistinguishable are called \emph{anonymous}.

It is known that many fundamental problems for anonymous networks cannot be solved without additional ``symmetry-breaking'' assumptions: A notable example is the \emph{Counting problem}, i.e., determining the total number of processes $n$. The most typical symmetry-breaking choice is assuming the presence of a single distinguished process in the system, called \emph{leader}~\cite{AAE08,ABBS16,BBCD15,BBK11,DFIISV19,FPP00,KM18,KM20,MCS13,Sa99,YK96}. A leader process may represent a base station in a sensor network, a super-node in a P2P network, etc.

\myskip
\mypar{Disconnected networks.} Another common assumption is that the network is connected at every round~\cite{KLO10, DW05}. However, this assumption appears somewhat far-fetched when one considers the highly dynamic nature of some real-world networks, such as P2P networks of smart devices moving unpredictably. A weaker and more reasonable assumption is that the union of all the network's links across any $T$ consecutive rounds induces a connected (multi)graph on the processes~\cite{KMarxiv,OT09}. Such a network is said to be \emph{$T$-union-connected}, and $T\geq 1$ is its \emph{dynamic disconnectivity}~\cite{DVdisc}. Typically, it is assumed $T$ (or an upper bound thereof) to be known by the processes.

\myskip
\mypar{Congested networks.} Almost all previous research on anonymous dynamic networks pertains to models that impose no limit on the size of messages exchanged by processes~\cite{DB15,DV22,KM18,KM19,KM20,KM21,KM22,DW05,OT09}. Unfortunately, in most mobile networks, sending small-size messages is not only desirable but also a necessity; for example, in sensor networks, short communication times significantly increase battery life. A more realistic model assumes the network to be ``congested'' and limits the size of every message to $O(\log n)$ bits, where $n$ is the number of processes~\cite{PD00}.\footnote{This $O(\log n)$ limit on message sizes does not imply that the processes have a-priori information about $n$. The size limit is not explicitly given to the processes, and it is up to the algorithm to automatically prevent larger messages from being sent.}

\myskip
\mypar{General computation.} A recent innovation in the study of anonymous dynamic networks with leaders was the introduction of \emph{history trees} in~\cite{DV22}, which led to an optimal deterministic solution to the \emph{Generalized Counting problem}\footnote{In the Generalized Counting problem, each process starts with a certain input, and the goal is to determine how many processes have each input. That is, each process has to compute the \emph{multiset} of all inputs.} in the non-congested network model. This problem is ``complete'' for a large class of functions called \emph{multiset-based functions}, which in turn are the only computable functions in this model. The theory of history trees was extended in~\cite{DVdisc} to \emph{leaderless} networks, providing optimal algorithms for the \emph{Frequency problem}:\footnote{In the Frequency problem, the goal is to determine the percentage of processes that have each input.} This problem is complete for the class of \emph{frequency-based functions}, which are the only computable functions in leaderless systems. Thus, the computational landscape for the non-congested network model is fully understood, and optimal linear-time algorithms are known for anonymous dynamic systems with and without leaders.\footnote{By the word ``optimal'' we mean ``asymptotically worst-case optimal as a function of the total number of processes $n$''.} No previous research exists on anonymous networks in the congested model, except for a recent preprint that gives a Counting algorithm in $\widetilde{O}(n^{2T(1+\epsilon)+3})$ rounds for networks with leaders~\cite{KMarxiv}. Note that its running time is exponential in the dynamic disconnectivity $T$ and becomes $\widetilde{O}(n^{5+\epsilon})$ for connected networks.

\subsection{Contributions and Techniques}\label{s:1.1}

\mypar{Summary.} In this paper, we provide a state-of-the-art general algorithmic technique for $T$-union-connected anonymous dynamic \emph{congested} networks, with and without leaders. The resulting algorithms run in $O(Tn^3)$ rounds, where $n$ is the (initially unknown) total number of processes, and the dynamic disconnectivity $T$ is initially known.

In \cref{s:3} we give a basic and slightly inefficient Counting algorithm that applies to a limited setting but already contains all of the key ideas of our technique. In \cref{s:4} we prove its correctness, and in \cref{s:5} we discuss optimizations and extensions of the basic algorithm to several other settings. \textbf{This includes the computation of all multiset-based functions.} \cref{s:6} concludes the paper with some directions for future research.

Preliminary conference versions of this paper appeared at PODC~2023~\cite{prev01} (as a Brief Announcement) and MFCS~2024~\cite{prev02}. This version contains additional technical details and all missing proofs.

\myskip
\mypar{Technical background.} Informally, a \emph{history tree} is a way of representing the history of a network in the form of an infinite tree. Each node in a history tree represents a set of anonymous processes that are ``indistinguishable'' at a certain round, where two processes become ``distinguishable'' as soon as they receive different sets of messages (see \cref{s:2} for proper definitions).

The theory of history trees developed in~\cite{DV22, DVdisc} yields optimal general algorithms for anonymous dynamic networks with and without leaders, assuming the network is not congested. The idea is that processes can work together to incrementally construct the history tree by repeatedly exchanging and merging together their respective ``views'' of it. Once they have a sufficiently large portion of the history tree (i.e., a number of ``levels'' proportional to $n$), each process can locally analyze its structure and perform arbitrary computations on the multiset of input values originally assigned to the processes.

\myskip
\mypar{Challenges.} Unfortunately, implementing the above idea requires sending messages containing entire ``views'' of the history tree. The size of a view is $\Theta(n^3\log n)$ bits in the worst case, and is therefore unsuitable for the congested network model~\cite{DV22}. There is a major difficulty in dealing with this problem deterministically, which stems from the lack of unique IDs combined with the dynamic nature of the network.

It is worth noting that the ``naive'' approach of breaking down large messages into smaller pieces to be sent in different rounds does not work. Indeed, it is not clear how the original message can then be reconstructed, because the pieces carry no IDs and a process' neighbors may change at every round. This may result in messages from different processes being mixed up and non-existent messages being reconstructed.

\myskip
\mypar{Methodology.} Our main contribution is a general method that allows history trees to be transmitted reliably and deterministically between anonymous processes in a dynamic \emph{congested} network with a leader. To overcome the fundamental issues outlined above, we devised a basic protocol combining different techniques, as well as a number of extensions, including leaderless ones. Although the techniques introduced in this paper are self-contained and do not rely on the results of~\cite{DV22}, they effectively allow us to reduce the congested network model to the non-congested one, making it possible to apply the Counting algorithm in~\cite{DV22} as a ``black box''.

Firstly, we developed a method for dynamically assigning temporary (non-unique) IDs to processes; this method is an essential part of the history tree transmission algorithm. In fact, the nodes of our history trees are now augmented with IDs, meaning that each node represents the set of processes with a certain ID. When processes with equal IDs get disambiguated, they get new IDs.

The transmission of history trees occurs level by level, one edge at a time. Since the total ordering between IDs induces a total ordering on the history tree's edges, the processes can collectively transmit sets of edges with a method reminiscent of \emph{Token Dissemination}~\cite{KLO10}.

Essentially, all processes participate in a series of broadcasts; the goal of each broadcast is to transmit the next ``highest-value'' edge to the whole network. The problem is that no upper bound on the \emph{dynamic diameter} of the network is known, and therefore there is no way of knowing how many rounds it may take for all processes to receive the edge being broadcast.

We adopt a self-stabilizing approach to ensure that all messages are successfully broadcast. We give a communication protocol based on acknowledgments by the leader, where failure to broadcast a message alerts at least one process. Alerted processes start broadcasting error messages, which eventually cause a reset of the broadcast that caused the error. A mechanism that dynamically estimates the diameter of the network guarantees that no more than $O(\log n)$ resets are performed.

Finally, in order to achieve a cubic running time, we do not construct the history tree of the actual network, but a more compact history tree corresponding to a \emph{virtual network}. The virtual network is carefully derived from the real one in such a way as to amortize the number of edges in the resulting history tree and further reduce the final worst-case running time by a factor of $n$.

\subsection{Previous Work}\label{s:1.2}

\mypar{Non-congested networks.} The study of computation in anonymous dynamic networks has been mainly devoted to two fundamental problems: The \emph{Counting problem} in networks with a leader~\cite{DB15,DBBC14b,DBCB13,KM18,KM19,KM20,KM21,KM22} and the \emph{Average Consensus problem} in leaderless networks~\cite{BT89,CL18,CL22,C11,KM21,NOOT09,O17,OT11,T84,YSSBG13}. This research thread produced a series of algorithms for these problems; the underlying technique used is a local averaging or mass-distribution method coupled with refined termination strategies.

A radically different technique based on history trees was recently used to optimally compute arbitrary multiset-based functions in $3n$ rounds in networks with a leader~\cite{DV22}. This approach was successfully extended to multi-leader, leaderless, and disconnected networks~\cite{DVdisc}.

\myskip
\mypar{Congested networks.} As for the congested model, the only paper that has ever studied deterministic algorithms for anonymous dynamic networks, to the best of our knowledge, is the recent preprint~\cite{KMarxiv}, which solves the \emph{Generalized Counting problem} in $\widetilde{O}(n^{2T(1+\epsilon)+3})$ rounds. As usual, $T$ is the (initially known) dynamic disconnectivity of the network and $\epsilon$ is an arbitrarily small positive constant. By comparison, our main algorithm has a running time of $O(Tn^3)$ rounds; hence, its dependence on $T$ is linear (as opposed to exponential) and, for connected networks (i.e., when $T=1$), the improvement is a factor of $\Theta(n^{2+\epsilon}\log^k n)$.

Most previous research efforts on congested networks in the dynamic setting have focused on randomized algorithms or processes with unique IDs~\cite{DPR13,HF12,JY24,KLO10}. In this context, a problem similar to Generalized Counting is \emph{Token Dissemination}, where each process starts with a token and the goal is for every process to collect all tokens. In connected networks, this problem is solved in $O(n^2)$ rounds by a simple \emph{token-forwarding} algorithm (i.e., no manipulation is done on tokens other than storing, copying, and individually transmitting them)~\cite{KLO10}. Interestingly, solving the Token Dissemination problem by token-forwarding algorithms requires at least $\Omega(n^2/\log n)$ rounds~\cite{DPR13}.

It is worth remarking that the randomized algorithm in~\cite{KLO10} only solves the Counting problem approximately and assumes a-priori knowledge of an upper bound on $n$. Moreover, this algorithm only works with high probability. These are three key differences that make our contribution preferable. Furthermore, assuming processes to have unique IDs (or randomly generating unique IDs) as in~\cite{KLO10} defeats the purpose of safeguarding user privacy, which is a motivation of our work.

\section{Definitions and Fundamentals}\label{s:2}

\mypar{Computation model.} A \emph{dynamic network} is modeled by an infinite sequence $\mathcal G=(G_t)_{t\geq 1}$, where $G_t=(V,E_t)$ is an undirected multigraph whose vertex set $V=\{p_1, p_2, \dots, p_n\}$ is a system of $n$ \emph{anonymous processes} and $E_t$ is a multiset of edges representing \emph{links} between processes. Hence, there may be multiple links between two processes, or even from a process to itself.\footnote{Each self-loop in $G_t$ represents a single link, hence a single message being sent and received by the same process.}

If there is a constant $T\geq 1$ such that, for every $t\geq 1$, the multigraph $G^\star_t=\left(V, \bigcup_{i=t}^{t+T-1} E_i\right)$ is connected, the network is said to be \emph{$T$-union-connected}, and the smallest such $T$ is its \emph{dynamic disconnectivity}.\footnote{A similar parameter for dynamic networks is the \emph{dynamic diameter} $D$, defined as the maximum number of rounds it may take for a message to be broadcast from a process to all other processes. The parameters $T$ and $D$ are related by the inequalities $T \leq D \leq T(n-1)$, which are best possible~\cite{DVdisc}.}

Each process $p_i$ starts with an \emph{input}, which is assigned to it at \emph{round~$0$}. It also has an internal state, which is initially determined by its input. At each \emph{round~$t\geq 1$}, every process composes a message (as a function of its internal state) and sends it to its neighbors in $G_t$ through all its incident links.\footnote{Contrary to static networks, where processes may be allowed to send different messages to different neighbors, dynamic networks usually require processes to ``broadcast'' the same message through all incident links due to the lack of unique port numbers. This is the case, for instance, in wireless radio communications, where messages are sent to all processes within communication range, and the anonymity of the network prevents destinations from being specified.} In the congested network model, only messages of $O(\log n)$ bits can be sent.\footnote{We may want to assume that the total number of links in a network multigraph $G_t$ (counted with their multiplicities) is bounded by a polynomial in $n$. In fact, when dealing with congested networks, we explicitly make this assumption, as it ensures that the multiplicity of a link can always fit in a single $O(\log n)$-sized message.} By the end of round~$t$, each process reads all messages coming from its neighbors and updates its internal state according to a local algorithm $\mathcal A$. Note that $\mathcal A$ is the same for all processes, and is a deterministic function of the internal state and the multiset of messages received in the current round.

The input of each process also includes a \emph{leader flag}. In \cref{s:3}, we will assume that the leader flag of exactly one process is set (this process is the \emph{unique leader}); in \cref{s:5}, we will discuss the \emph{leaderless} case, where none of the processes has the leader flag set.

A process may return an \emph{output} at the end of a round, which must be a function of its current internal state. A process may also \emph{terminate} execution after returning an output. An algorithm $\mathcal A$ solves the \emph{Counting problem} if executing $\mathcal A$ at every round eventually causes all processes to simultaneously output $n$ and terminate.\footnote{Note that some definitions of the Counting problem found in the literature do not require all processes to terminate at the same time. In fact, in \cref{s:3}, for the sake of simplicity, we will consider the Counting problem solved as soon as the leader outputs $n$. We will discuss the full-fledged Counting problem, as well as the computation of arbitrary functions, in \cref{s:5}.} The (worst-case) \emph{running time} of $\mathcal A$, as a function of $n$, is the maximum number of rounds it takes for $\mathcal A$ to solve the problem, across all possible dynamic networks of size $n$ and all possible input assignments.

\begin{figure}
\centering
\includegraphics[scale=0.5]{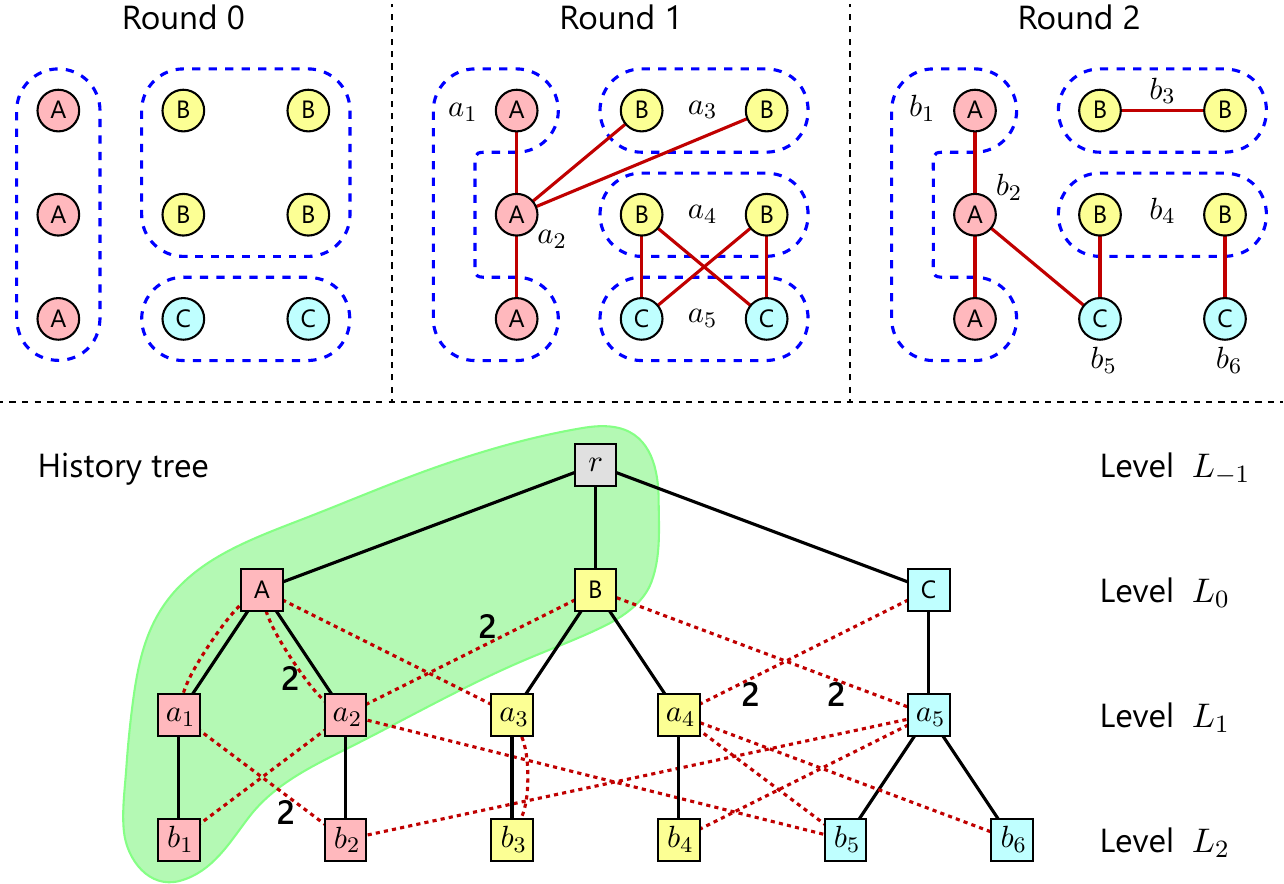}
\caption{The first rounds of a dynamic network with $n=9$ processes and the corresponding levels of the history tree. Level $L_t$ consists of the nodes at distance $t+1$ from the root $r$, which represent indistinguishable processes after the $t$th communication round. There are no leaders in the network, but each process has an input from the set $\{A,B,C\}$. Only the nodes in $L_0$ have explicit labels; all labels of the form $a_i$ and $b_i$ were added for the reader's convenience, and indicate classes of indistinguishable processes (in contrast, the nodes of the \emph{virtual history tree} introduced in \cref{s:3} do have IDs). Note that the two processes in $b_4$ are still indistinguishable at the end of round~$2$, although they are linked to the distinguishable processes $b_5$ and $b_6$. This is because such processes were in the same class $a_5$ at round~$1$. The subgraph induced by the nodes in the green blob is the \emph{view} of the two processes in $b_1$. We remark that a history tree does not contain any explicit information about how many processes each node represents.}
\label{fig:ht1}
\Description{Example of a history tree}
\end{figure}

\myskip
\mypar{History trees.} \emph{History trees} were introduced in~\cite{DV22} as a tool of investigation for anonymous dynamic networks; an example is found in \cref{fig:ht1}. A history tree is a representation of a dynamic network given some inputs to the processes. It is an infinite graph whose nodes are partitioned into \emph{levels} $L_t$, with $t\geq -1$; each node in $L_t$ represents a class of processes that are \emph{indistinguishable} at the end of round~$t$ (with the exception of $L_{-1}$, which contains a single node $r$ representing all processes). The definition of distinguishability is inductive: At the end of round~$0$, two processes are distinguishable if and only if they have different inputs. At the end of round~$t\geq 1$, two processes are distinguishable if and only if they were already distinguishable at round~$t-1$ or if they have received different multisets of messages at round~$t$. (We refer to ``multisets'' of messages, as opposed to sets, because multiple copies of identical messages may be received; each message has a \emph{multiplicity}.)

Each node in level $L_0$ has a label indicating the input of the processes it represents. There are also two types of edges connecting nodes in adjacent levels. The \emph{black edges} induce an infinite tree spanning all nodes, rooted at node $r\in L_{-1}$. The presence of a black edge $\{v, v'\}$, with $v\in L_{t}$ and $v'\in L_{t+1}$, indicates that the \emph{child node} $v'$ represents a subset of the processes represented by the \emph{parent node} $v$. The \emph{red multi-edges} represent communications between processes. The presence of a red edge $\{v, v'\}$ with multiplicity $m$, with $v\in L_{t}$ and $v'\in L_{t+1}$, indicates that, at round~$t+1$, each process represented by $v'$ receives $m$ (identical) messages from processes represented by $v$.

The \emph{view} of a process $p$ at round~$t\geq0$ is the subgraph of the history tree that is spanned by all the shortest paths (using black and red edges indifferently) from the root $r$ to the node in $L_t$ representing $p$ (\cref{fig:ht1} shows an example of a view).

\myskip
\mypar{Applications of history trees.} As proved in~\cite{DV22}, in a connected network with a unique leader, the view of any process at round $3n$ contains enough information to determine the multiset of the initial inputs (i.e., how many processes have each input value). Thus, once a process is able to locally construct a view spanning $3n$ levels of a history tree, it can immediately do any computation on the inputs, and in particular determine $n$ and solve the Generalized Counting problem.

History trees were adopted in~\cite{DV22,DVdisc}, where it is shown how processes can construct their views of the history tree in real time. For the algorithm to work, each process is required to repeatedly send its current view to all its neighbors at every round, merging it with all the views it receives from them.

This approach is not feasible in the congested network model, because a view at round $t$ has size $\Theta(tn^2\log n)$ in the worst case, since there may be $\Theta(n^2)$ red edges in each level. In the following, we will develop a strategy whereby processes can construct a history tree one red edge at a time. The core idea is that the nodes of the history tree are assigned unique IDs, and therefore a single red edge can be encoded in only $O(\log n)$ bits as a pair of IDs and a multiplicity.

\section{Basic Counting Algorithm}\label{s:3}

In this section we describe our deterministic Counting algorithm for congested anonymous dynamic networks in its most basic version. The algorithm assumes the network to be connected ($T=1$), to have a unique leader, and execution terminates in $O(n^3\log n)$ rounds with the leader reporting the total number of processes in the system, $n$. The complete pseudocode is found in \cref{a:1}.

In \cref{s:5}, we will show how to optimize this basic algorithm, making it terminate in $O(n^3)$ rounds. We will also extend the algorithm in several directions, for instance by making all processes (as opposed to the leader only) simultaneously output $n$ and terminate. Furthermore, we will show how to not only count the number of processes, but also compute arbitrary (multiset-based, cf.~\cite{DV22}) functions, assuming that input values are assigned to the processes. Finally, we will extend the algorithm to leaderless networks and $T$-union-connected networks.

\subsection{Algorithm Outline}\label{s:3.1}

The only input given to each process is whether or not it is the leader; this information can be retrieved by calling the primitive \texttt{Input}, as in \al{1}{15}. Each process also has some private memory which is used to permanently store information in the form of internal variables (\cref{l:1}).

\myskip
\mypar{Virtual history tree (VHT).} The overall goal of the algorithm is for the processes to implicitly agree on the first $O(n)$ levels of a particular history tree, called \emph{virtual history tree (VHT)}, which corresponds to a dynamic network $\mathcal N$ of $n$ processes. Once the construction of each new level of the VHT is complete (refer to \cref{s:3.5}), the leader locally runs the Counting algorithm from~\cite{DV22} on the VHT (\als{2}{29--31}). If this algorithm successfully returns a number (as opposed to ``Unknown''), the leader outputs it; otherwise, the construction of a new level of the VHT is initiated.

\myskip
\mypar{Virtual network ($\mathcal N$).} The dynamic network $\mathcal N=(N_1, N_2, \dots)$ represented by the VHT is in fact a \emph{virtual network}, in the sense that none of the multigraphs $N_t$ necessarily coincides with any multigraph of links actually occurring in the real communication network $\mathcal G=(G_1, G_2, \dots)$. However, each $N_t$ is obtained by carefully adding and removing links from some $G_{i_t}$ (see \cref{fig:ht2}). This manipulation has the purpose of reducing the size of the resulting VHT by a factor of $n$ (see \cref{s:3.4,l:vhtsize}).

\myskip
\mypar{Temporary IDs.} To cope with the fact that processes are anonymous and information can only be sent in small chunks of size $O(\log n)$, each process has a \emph{temporary ID} stored in a local variable called \texttt{MyID}. Each node $v$ in the VHT also has an ID, indicating that $v$ represents all processes having that ID. Thus, a \emph{red-edge triplet} of the form $(\texttt{ID1}, \texttt{ID2}, \texttt{Mult})$ can be used to unambiguously represent a red edge of multiplicity \texttt{Mult} between the nodes of the VHT whose IDs are \texttt{ID1} and \texttt{ID2}. Since a red-edge triplet has size $O(\log n)$, it can be included in a single message. Note that the variable \texttt{MyID} of each process may be modified over time as the VHT acquires more nodes.

\myskip
\mypar{Broadcast phases.} The construction of the VHT is carried out level by level, and is done through several \emph{broadcast phases}, which are indirectly coordinated by the leader (see \cref{s:3.3}). At first, each process knows the red edges incident to its corresponding node of the VHT. Then, ideally every two broadcast phases, the whole network learns a new red edge of the VHT. The broadcast phases continue until all processes know all red edges in the level (see \cref{s:3.6}).

\myskip
\mypar{Estimating the diameter.} In order to guarantee the success of a broadcast phase, all processes must keep sending each other information for a certain number of rounds, which depends on the \emph{dynamic diameter} of the network, and is $n-1$ in the worst case~\cite{KLO10}. Since the processes do not initially possess any information at all, they can only make estimations on the dynamic diameter. The current estimate is stored by each process in the variable \texttt{DiamEstimate}. Its value is initially $1$, and it is doubled every time the processes detect an error in a broadcast.

\myskip
\mypar{Error phases.} Detecting broadcasting errors and consistently reacting to them is by no means a trivial task, and is discussed in \cref{s:3.7}. Our broadcasting technique ensures that, if some red-edge triplet fails to be broadcast to the entire network and does not become part of the local VHT of all processes, at least one process becomes aware of this fact. Such a process enters an \emph{error phase}, sending a high-priority message at every round containing the level number at which the error occurred. Error messages supersede the regular ones and eventually reach the leader.

\myskip
\mypar{Reset phases.} When the leader finally receives an error message, it initiates a \emph{reset phase}, whose goal is to force the whole network to restore a previous state of the VHT and continue from there. This is achieved by broadcasting a high-priority reset message. Since the error must have occurred because \texttt{DiamEstimate} was too small, its value is doubled at the end of the reset phase.

Note that there is no obvious way for the leader to tell if any level of the VHT is actually missing some parts. Indeed, at any point in time, there may be processes in an error phase unbeknownst to the leader. One of the challenges of our method is to ensure that the leader will not terminate with an incorrect guess on $n$ due to the VHT being incomplete (see \cref{t:main}).

\subsection{Communication and Priority}\label{s:3.2}

\mypar{Counting rounds.} The processes have to implicitly synchronize with one another to start and finish each broadcast phase at the same time. Part of the synchronization is achieved by the function \texttt{SendAndReceive} (\als{3}{14--18}), which is called by each process at every communication round. This function simply sends a given message to all neighbors, collects all messages coming from the neighbors, and increments the internal variable \texttt{CurrentRound}. Since communications in the network are synchronous, all processes always agree on the value of \texttt{CurrentRound}.

\myskip
\mypar{Message types.} The processes use messages of various types to share information with one another. Each message has a \emph{label} describing its type, as well as at most three additional integer parameters. As it will turn out in the analysis of the algorithm, each parameter has size $O(\log n)$ bits (see \cref{c:msize}).\footnote{As already remarked, the algorithm spontaneously creates $O(\log n)$-sized messages without any a-priori knowledge on $n$.}

The message types and their parameters are as follows.

\begin{itemize}
\item \textbf{Null message.} Label: ``Null''. No parameters.
\item \textbf{Level-begin message.} Label: ``Begin''. Parameters: \texttt{ID}.
\item \textbf{Level-end message.} Label: ``End''. No parameters.
\item \textbf{Done message.} Label: ``Done''. Parameters: \texttt{ID}.
\item \textbf{Red-edge message.} Label: ``Edge''. Parameters: \texttt{ID1}, \texttt{ID2}, \texttt{Mult}.
\item \textbf{Error message.} Label: ``Error''. Parameters: \texttt{ErrorLevel}.
\item \textbf{Reset message.} Label: ``Reset''. Parameters: \texttt{ResetLevel}, \texttt{StartingRound}, \texttt{NewDiam}.
\end{itemize}

\myskip
\mypar{Priority.} Messages have \emph{priorities} that determine how they are handled during a broadcast. The priority of a message is defined in \als{3}{1--12}, and can be summarized as follows:
\begin{center}
Null $<$ Begin $<$ End $<$ Done $<$ Edge $<$ \dots $<$ Reset $k+1$ $<$ Error $k$ $<$ Reset $k$ $<$ \dots $<$ Error 1 $<$ Reset 1
\end{center}
That is, the message with lowest priority is the Null message, followed by all possible Level-begin messages, then the Level-end message, etc. The priority of a Level-begin message is independent of its parameter. For all other message types, however, the priority is also a function of the parameters.

As for Done messages, different \texttt{ID} parameters yield different priorities. Thus, priority induces a total ordering on the set of all possible Done messages; the precise ordering is irrelevant, as long as all processes implicitly agree on it. All Done messages have greater priority than Null, Level-begin and Level-end messages. The same goes for Red-edge messages: Different parameters yield different priorities, and all processes agree on the priority function. All Red-edge messages have greater priority than all Done messages, and lower priority than all Error and Reset messages.

An Error message (respectively, a Reset message) with a smaller \texttt{ErrorLevel} (respectively, \texttt{ResetLevel}) has a greater priority. Moreover, the priorities of Error and Reset messages are interleaved: The priority of an Error message with $\texttt{ErrorLevel}=k$ is strictly between the priority of a Reset message with $\texttt{ResetLevel}=k+1$ and the priority of a Reset message with $\texttt{ResetLevel}=k$.

\subsection{Broadcasting Data}\label{s:3.3}

During the execution of the algorithm, a non-leader process may have a particular piece of information that it wishes to send to the leader. Similarly, the leader may have some information that it wishes to share with all processes in the network. Both operations are performed via a \emph{broadcast} spanning several rounds.

The broadcast technique used to construct the VHT is implemented in \als{3}{20--38}. It is assumed that all processes participating in the broadcast are \emph{synchronized}, i.e., they start at the same round and continue broadcasting for the same number of rounds (which is equal to \texttt{DiamEstimate}, cf.~\als{3}{30--31}). Each process is also assumed to have the information it wishes to share packed in a message of the appropriate type (see \cref{s:3.6}), which is passed to the function \texttt{BroadcastPhase} as the argument \texttt{Message}.

At each broadcast round, each process sends its message to all its neighbors. Then it examines the messages received from the neighbors, as well as its own message, and keeps only the message with highest priority, discarding all others (function \texttt{BroadcastStep}). This is the message the process will send in the next round, and so on.\footnote{Note that our broadcasting strategy implements a token-forwarding mechanism, and the sequence of all broadcast phases is akin to a Token Dissemination algorithm. It is known that any such algorithm must have a worst-case running time of at least $\Omega(n^2/\log n)$ rounds~\cite{DPR13}.}

Ideally, if the broadcast is continued for a sufficiently large number of rounds, all processes participating in the broadcast will eventually obtain the message having the highest priority among the ones initially owned by the processes (note that this may be an Error or Reset message, as well).

\subsection{Defining the Virtual Network}\label{s:3.4}

\mypar{Definition.} Recall that the VHT is the history tree of the virtual network $\mathcal N=(N_1, N_2, N_3, \dots)$. We will now define $N_t$ by induction on $t$. That is, assuming that the multigraphs $N_1$, $N_2$, \dots, $N_{t-1}$ are already known, we will construct $N_t$ based on the multigraph $G_{i_t}$, which represents the real communication network $\mathcal G$ at a selected round $i_t$ (\cref{s:3.6} explains how $i_t$ is selected).

Since the first $t-1$ rounds of the virtual network are known, we can construct the first levels of the VHT up to level $L_{t-1}$. By definition of history tree, each node $v\in L_{t-1}$ represents a class $P_v$ of processes that are indistinguishable after the first $t-1$ ``virtual rounds'' modeled by the communication networks $N_1$, $N_2$, \dots, $N_{t-1}$.

Consider the simple undirected graph $H=(V, E)$, where $V=L_{t-1}$ and $\{u, v\}\in E$ if and only if $u\neq v$ and there is at least one link in $G_{i_t}$ between a process in $P_u$ and one in $P_v$. Recall that we are assuming that $G_{i_t}$ is connected, and therefore $H$ is connected, as well (disconnected networks will be discussed in \cref{s:5}). Let $S=(V, E')$ be an arbitrary spanning tree of $H$.

Now, $N_t$ is the network having the following links (refer to \cref{fig:ht2}):
\begin{itemize}
\item For every edge $\{u,v\}\in E'$ of $S$, the multigraph $N_t$ contains all the links in $G_{i_t}$ having an endpoint in $P_u$ and an endpoint in $P_v$ (with the respective multiplicities).
\item For every $v\in V$, the multigraph $N_t$ contains a cycle $C_v$ spanning all the processes in $P_v$. In the special case where $P_v$ contains exactly two processes $p_1$ and $p_2$, $C_v$ is a double link between $p_1$ and $p_2$. If $P_v$ contains a single process $p$, $C_v$ is a double self-loop on $p$.
\end{itemize}
Note that, in every case, $C_v$ induces a 2-regular multigraph on $P_v$. The purpose of these (possibly degenerate) cycles $C_v$ is to ensure that $N_t$ is connected. On the other hand, the purpose of using a spanning tree $S$, as opposed to the full graph $H$, is to reduce the size of the VHT (see \cref{l:vhtsize}).

\myskip
\mypar{Implementation.} In the actual algorithm, $N_t$ is defined only implicitly in a distributed manner (because the ultimate goal is merely the construction of the VHT). The implicit construction of $N_t$ starts at round $i_t$, where each process (not in an error phase) invokes the function \texttt{SetUpNewLevel} and sends a Begin message containing its own ID to each neighbor (\als{4}{2--5}). As a result, each process learns the IDs of all its neighbors in $G_{i_t}$, as well as their multiplicities, and stores this information as a list of pairs of the form $(\texttt{ID}, \texttt{Mult})$ in the internal variable \texttt{ObsList} (\als{4}{6--12}). It should be noted that a process discards all Begin messages from processes with the same ID as its own. These are replaced by the single pair $(\texttt{MyID}, 2)$ (\al{4}{13}), which accounts for the two edges of the cycle $C_v$ incident to the process in $N_t$.

The choice of links to be included in $N_t$ (which directly reflects on the red edges included in the VHT) is guided by the construction and maintenance of an \emph{auxiliary graph} stored in each process' internal variable \texttt{LevelGraph}. The auxiliary graph is a graph on $V=L_{t-1}$ and starts with no edges at all (\als{4}{15--18}); it gradually acquires more edges until it becomes the spanning tree $S$ (as defined above). This is carried out as part of the function \texttt{UpdateTempVHT}, which is invoked every time a red-edge triplet $(\texttt{ID1}, \texttt{ID2}, \texttt{Mult})$ is selected to become a new red edge of the VHT at the end of a broadcast phase. Among other operations (described in \cref{s:3.5}), this function adds an edge to the auxiliary graph which connects the two nodes corresponding to \texttt{ID1} and \texttt{ID2} (\als{5}{29--32}). Then the function \texttt{PreventCyclesInLevelGraph} is called, which deletes all the pairs in \texttt{ObsList} whose selection would cause the creation of a cycle in the auxiliary graph (\als{5}{7--15}). This guarantees that, eventually, \texttt{LevelGraph} will be a tree (representing $S$).

\begin{figure}
\centering
\includegraphics[width=\linewidth]{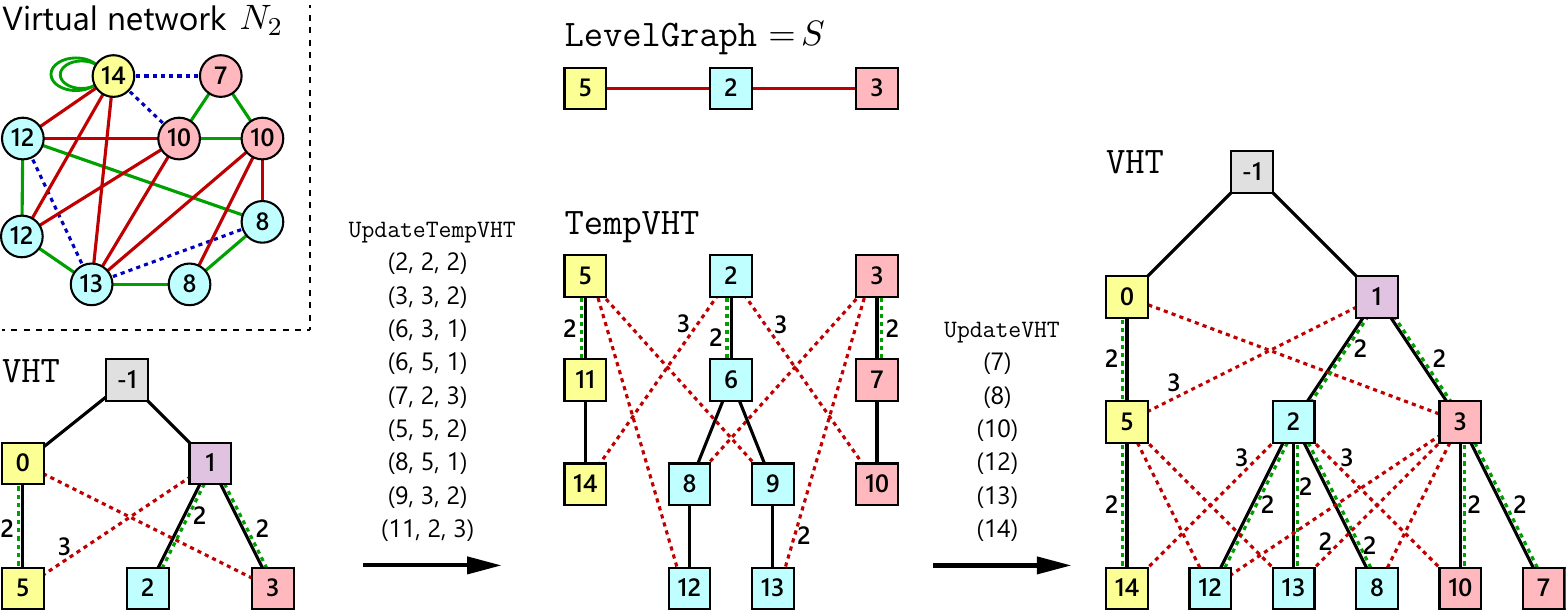}
\caption{Construction of the virtual network $N_2$ and level $L_2$ of the virtual history tree (VHT). The real network $G_{i_2}$ consists of the $n=9$ processes in the upper-left picture, connected by the red and blue links (not the green ones). Before the construction starts, the IDs in the network are $5$, $2$, and $3$. Same-colored processes have equal IDs when the construction of $L_2$ starts; the labels indicate their IDs when the construction finishes. Accordingly, level $L_1$ of the VHT has three nodes with IDs $5$, $2$, $3$. The graph $H$ is a triangle on these three nodes, while its spanning tree $S$ is the same as \texttt{LevelGraph}. Therefore, to construct $N_2$ we remove the blue edges from $G_{i_2}$ and keep the red ones (see \cref{s:3.4}). For example, after the triplets $(6,3,1)$ and $(6,5,1)$ are accepted, any elements of \texttt{ObsList} corresponding to the two blue edges incident to the yellow process are deleted and are never broadcast. We also have to add the green edges, which represent the cycles $C_v$. Note that failing to add them would result in a disconnected network, because the process with ID~$7$ would be isolated. For clarity, the edges of the (temporary) VHT are colored red or green to match the edges of the virtual network that they represent (although technically they are all red edges).}
\label{fig:ht2}
\Description{Construction of the VHT}
\end{figure}

\subsection{Constructing the Virtual History Tree (VHT)}\label{s:3.5}

\mypar{Initialization.} The VHT is initialized by the function \texttt{InitializeVariables} in \cref{l:1}. At first, the VHT only contains level $L_{-1}$ (a single root node whose ID is $-1$) and level $L_0$ (two nodes with IDs $0$ and $1$, representing the leader and the non-leader processes, respectively). Accordingly, the leader initializes its own \texttt{MyID} variable to $0$ and the non-leaders to $1$.

\myskip
\mypar{Temporary VHT.} Then, for all $t\geq 1$, the level $L_t$ of the VHT is constructed based on the virtual network $N_t$ as in \cref{fig:ht2}. When the construction begins, all processes (that are not in an error phase) acquire new pairs of the form $(\texttt{ID}, \texttt{Mult})$ at the same round $i_t$ and store them in \texttt{ObsList}, as explained in \cref{s:3.4}. The construction of a new level is not carried out directly on the VHT, but in a separate \emph{temporary VHT} stored in the internal variable \texttt{TempVHT}. Initially, this is just a copy of level $L_{t-1}$ of the VHT (\als{4}{14--17}), i.e., a set of nodes, each with a distinct ID, and no edges.

\myskip
\mypar{Adding red edges.} In order to determine how \texttt{TempVHT} should be updated, several broadcast phases are performed, allowing the processes to share their red-edge triplets with one another. Note that a process can transform each element of \texttt{ObsList} of the form $(\texttt{ID}, \texttt{Mult})$ into a red-edge triplet by simply adding its own ID as the first element (\als{4}{30--34}). Every time a red-edge triplet $(\texttt{ID1}, \texttt{ID2}, \texttt{Mult})$ has been broadcast by a process to the entire network and has been ``accepted'' (see \cref{s:3.6} for details), the function \texttt{UpdateTempVHT} is called (\als{2}{24--25}). The result is that the node $v$ of \texttt{TempVHT} whose ID is \texttt{ID1} gets a new child $v'$ with a new unique ID (\als{5}{21--24}). Pictorially, a black edge is created connecting $v'$ to $v$. Also, a red edge with multiplicity \texttt{Mult} is added to \texttt{TempVHT}, connecting the new node with the node whose ID is \texttt{ID2} (\al{5}{25}).

\myskip
\mypar{Updating IDs.} By definition, the red-edge triplet $(\texttt{ID1}, \texttt{ID2}, \texttt{Mult})$ indicates that \emph{some} processes whose ID is \texttt{ID1} have received exactly \texttt{Mult} messages from processes whose ID is \texttt{ID2}. These processes, which were previously represented by the node $v$, are now represented by $v'$. Therefore, every process whose ID is \texttt{ID1} that has the pair $(\texttt{ID2}, \texttt{Mult})$ in its local \texttt{ObsList} removes it from the list and modifies its own ID from \texttt{ID1} to the ID of $v'$ (\als{5}{26--28}).\footnote{To clarify, the nodes of the VHT have unique IDs. Each node represents a class of "indistinguishable" processes, all of which have the same ID as the node itself. When processes disambiguate in the virtual network, they obtain different IDs. However, it is not necessarily true that all processes will have distinct IDs eventually. For example, if the network is the (static) complete graph, all non-leader processes will always have the same ID (which is incremented at every virtual round).}

\myskip
\mypar{Updating the VHT.} When a process has no more red-edge triplets to share (that is, its \texttt{ObsList} is empty), it broadcasts a Done message containing its current ID (\als{4}{27--28}). Recall that Done messages have lower priority than Red-edge messages (see \cref{s:3.2}). When the final result of a broadcast is a Done message containing a certain ID, all processes assume that \emph{some} processes with that ID have sent all their red-edge triplets. Therefore, the node of the temporary VHT with that ID is ready to be added to the VHT.

To this end, the function \texttt{UpdateVHT} is called with the ID contained in the Done message (\als{2}{26--27}). This function creates a new node $v$ in level $L_t$ of \texttt{VHT} corresponding to the node $v'$ of \texttt{TempVHT} whose ID is the one passed to the function (refer to \cref{fig:ht2}). The node $v$ gets the same ID as $v'$ and becomes a child of the node $u\in L_{t-1}$ having the same ID as the root $v''$ of the tree containing $v'$ (\als{5}{36--41}). Then, $v$ takes all the red edges found along the path from $v'$ to $v''$ (\als{5}{42--48}).

\myskip
\mypar{Finalizing the level.} When the \texttt{ObsList} of a process is empty and the VHT already contains a node with its ID, the process broadcasts an Level-end message, which has lower priority than Done and Red-edge messages (\als{4}{24--25}). When the result of a broadcast is an Level-end message, the construction of the level is finished.

\subsection{Main Loop}\label{s:3.6}

The entry point of the algorithm is the function \texttt{Main} in \cref{l:2}. After initializing the internal variables as already explained, the function goes through a loop that constructs the VHT level by level. At several points in this loop there may be errors that cause some of the levels to be undone and execution to resume from the beginning of the loop; we will discuss errors in \cref{s:3.7}. Next, we will describe an ideal error-free execution.

\myskip
\mypar{Level initialization.} With each iteration of the main loop, a new level of the VHT is constructed. At the beginning, the function \texttt{SetUpNewLevel} is executed by all processes: This marks a selected round $G_{i_t}$ as defined in \cref{s:3.4} (the non-trivial fact that all processes always call this function at the same round will become apparent from our analysis of the algorithm, in \cref{s:4}).

\myskip
\mypar{VHT broadcast.} Then an inner loop is entered; the purpose of each iteration is for all processes to learn new information, which causes an update of the temporary VHT or the VHT itself.

At first, each process calls the function \texttt{MakeVHTMessage}, which picks an element from \texttt{ObsList}, converts it into a red-edge triplet by adding its own ID to it, and wraps it in a Red-edge message. If \texttt{ObsList} is empty, a Done message containing the process' ID is generated instead (\als{4}{21--35}).

The resulting message is used in a first broadcast phase which, after a number of rounds equal to \texttt{DiamEstimate}, returns the highest-priority message circulating in the network. This message is stored in the variable \texttt{VHTMessage} (\al{2}{12}).

\myskip
\mypar{Acknowledgment broadcast.} Note that, in the presence of a faulty broadcast, different processes may end up having different versions of \texttt{VHTMessage}. To ensure that all processes (that are not in an error phase) update their local copies of the (temporary) VHT in a consistent way, an ``acknowledgment'' broadcast phase is performed. Its purpose is for the leader to inform all other processes of the \emph{accepted message}; by definition, the accepted message is the leader's version of \texttt{VHTMessage}.

In the acknowledgment phase, all non-leader processes broadcast a low-priority Null message, while the leader broadcasts the accepted message. The message resulting from the broadcast is then stored in the variable \texttt{AckMessage} (\als{2}{14--19}).

\myskip
\mypar{Updating the level.} Now, each process compares the contents of the two messages \texttt{VHTMessage} and \texttt{AckMessage}. If they are the same (hence not Null), then this is indeed the accepted message coming from the leader, and the data therein is used to update the temporary VHT or the VHT. Specifically, if \texttt{AckMessage} is a Red-edge message, the red-edge triplet therein is used to update the temporary VHT (\als{2}{24--25}). If it is a Done message, the VHT is updated instead (\als{2}{26--27}).

\myskip
\mypar{Finalizing the level.} When a process is done broadcasting red-edge triplets and the VHT already contains a node representing it, the process broadcasts a low-priority Level-end message. As soon as the accepted message is the Level-end message, the level is considered complete and the inner loop is exited (\al{2}{28}).

\myskip
\mypar{Counting processes.} Now the leader extracts its own view from the VHT and locally runs the Counting algorithm from~\cite{DV22} on it (function \texttt{CountFromView}). By ``extracts its own view'' we mean that it makes a copy of the VHT and deletes all nodes that are not on a shortest path from the root to the deepest leader's node (the correctness of this operation is proved in \cref{c:vhtcorr}).

If the value returned by \texttt{CountFromView} is a number, this is taken as the correct number $n$ of processes in the virtual network $\mathcal N$ (and hence in the real network $\mathcal G$) and becomes the leader's output (\als{2}{29--31}). Otherwise, a new iteration of the main loop starts, a new level of the VHT is constructed, and so on.

\subsection{Handling Errors and Resets}\label{s:3.7}

\mypar{Detecting errors.} As mentioned, a broadcast phase lasts a number of rounds indicated by the variable \texttt{DiamEstimate}. If this number happens to be smaller than the dynamic diameter of the network, the broadcast may be unsuccessful, in the sense that not all processes may end up agreeing on the same highest-priority message. Fortunately, the protocol described in \cref{s:3.6} makes error detection very simple: At the end of every acknowledgment broadcast, if the contents of a process' variables \texttt{VHTMessage} and \texttt{AckMessage} are different, the broadcast was unsuccessful. Any process that detects this event enters an error phase (\als{2}{21--22}).

\myskip
\mypar{Error phase.} When a process enters an error phase, it runs the function \texttt{BroadcastError} (\als{6}{21--27}); as it turns out, only non-leader processes ever execute this function. During an error phase, a process continually broadcasts an Error message containing the index of the level of the VHT where the error was detected. This message is replaced by any other message of higher priority received by the process during the error phase (see \cref{s:3.2}). As soon as the process receives a Reset message for a level of smaller or equal depth than the one indicated by the current Error message (i.e., a Reset message of higher priority), it ends the error phase and enters a reset phase.

Another situation where a non-leader process enters an error phase is when it receives an Error message from some other process. In this case, it calls the function \texttt{HandleError} and starts broadcasting a new Error message containing the smaller between the index of the current level and the one contained in the received Error message (\als{6}{10--11}). This causes higher-priority error messages to propagate through the network and eventually reach the leader.

\myskip
\mypar{Reset phase.} This phase is initiated by the leader at the end of a broadcast phase (or after sending a Begin message) in case it received an Error message. The leader first waits for $2\cdot \texttt{DiamEstimate} +1$ rounds, sending only Null messages (\als{6}{12--16}). This is to ensure that all non-leader processes finish any broadcast phase and enter an error phase. In turn, this prevents any possible conflicts between different broadcast phases before and after a reset (see \cref{s:4}).

Then the leader creates a Reset message containing the index of the VHT level where the error occurred, the current round number, and the new estimate on the dynamic diameter, i.e., twice the current one, since the error occurred because the estimate was too small (\als{6}{1--7}).

Now the leader calls the function \texttt{BroadcastReset}, where it broadcasts the Reset message; any process that receives this message starts broadcasting it as well (provided that it has higher priority than the Error message it is currently broadcasting, as usual). The broadcast continues for a total number of rounds equal to the new dynamic diameter estimate. All processes that receive the Reset message are able to synchronize with one another and finish the reset phase at the same round, thanks to the information contained in the message itself (\als{6}{30--33}).

At the end of the reset phase, the reset is actually performed: Every process that got the Reset message deletes the most recent levels of the VHT up to the level where the error occurred (this information is contained in the Reset message) and reverts its ID to the one it had at the beginning of the construction of that level. The variable \texttt{DiamEstimate} is also updated with the new estimate (\als{6}{34--41}). At this point, the network is ready to resume construction of the VHT.

\section{Proof of Correctness}\label{s:4}

\mypar{Process phases.} In order to prove the correctness of the Counting algorithm in \cref{s:3}, it is convenient to formally define some terminology. We say that a process, at a certain round, is in a \emph{broadcast phase}, in an \emph{error phase}, or in a \emph{reset phase} if the process is executing, respectively, the function \texttt{BroadcastPhase}, \texttt{BroadcastError}, or \texttt{BroadcastReset}. More specifically, if the function \texttt{BroadcastPhase} was invoked from \al{2}{12}, the process is in an \emph{VHT broadcast phase}; if it was invoked from \al{2}{15 or~19}, the process is in an \emph{acknowledgment broadcast phase}. A process is in a \emph{begin round} if it runs the function \texttt{SetUpNewLevel} at that round. When the leader is sending Null messages in the loop at \als{6}{15--16}, it is in a \emph{wait phase}. We refer to processes in an error phase as \emph{error processes}, and to the others as \emph{non-error processes}.

\myskip
\mypar{Global agreement and phase separation.} We will prove that our broadcasting protocols guarantee that phases occur ``consistently'' among different processes. We say that two processes \emph{agree} on a variable at a certain round if their local instances of that variable are equal.

\begin{lemma}\label{l:diam}
At any round, all non-error processes agree on \texttt{DiamEstimate}.
\end{lemma}
\begin{proof}
We can prove it by induction on the number of reset phases that the leader has completed. All processes initialize the variable to $1$, and modify it only at the end of a reset phase. Assume that all non-error processes agree on $\texttt{DiamEstimate}=d$ at some round. Before the leader initiates a reset phase, it waits for $2d+1$ rounds, at the end of which all processes must be in an error phase. Indeed, this is the maximum time they may spend in a begin round and two broadcast phases before realizing that no acknowledgment came from the leader. Our reset protocol guarantees that all processes involved in the reset phase end it at the same round, updating their \texttt{DiamEstimate} to the same value. All other processes are still in an error phase, and the induction goes through.
\end{proof}

\begin{lemma}\label{l:phases}
If the leader is in a begin round (respectively, a VHT broadcast phase, an acknowledgment broadcast phase, a reset phase), then so are all non-error processes.
\end{lemma}
\begin{proof}
We have already proved the claim for reset phases in \cref{l:diam}. The other claims hold by induction, due to the fact that all non-error processes always agree on \texttt{DiamEstimate}, and are therefore able to implicitly synchronize all their phases with the leader. That is, since all non-error processes finish every reset phase at the same round, and since all broadcast phases have a duration of \texttt{DiamEstimate} rounds, they must execute all of these phases in unison with the leader.
\end{proof}

\begin{corollary}\label{c:totalagree}
At any round, all non-error processes agree on \texttt{CurrentLevel}, \texttt{VHT}, \texttt{TempVHT}, \texttt{LevelGraph}, and \texttt{NextFreshID}.
\end{corollary}
\begin{proof}
A straightforward inductive proof follows from \cref{l:phases} and the fact that all non-error processes update these variables at a begin round or at the end of a broadcast phase or reset phase. Indeed, the acknowledgment protocol guarantees that a Red-edge or a Done message is used to update the (temporary) VHT only if it is the one accepted by the leader; also, any process that does not update its (temporary) VHT at the same time as the leader becomes an error process.
\end{proof}

Thanks to \cref{l:phases}, we can now unambiguously reason about the phases of the entire network (minus the error processes), because these phases coincide with the phases of the leader. Similarly, we can refer to the variables mentioned in \cref{l:diam,c:totalagree} with no explicit reference to a (non-error) process, because they all coincide with the local variables of the leader.

\myskip
\mypar{Faulty broadcast phases.} Recall that, if no reset interrupts the construction of the VHT, two consecutive broadcast phases occur between each call to \texttt{UpdateTempVHT} or to \texttt{UpdateVHT}. In the VHT broadcast phase, each process broadcasts a message containing information about the VHT; in the acknowledgment broadcast phase, the leader broadcasts the accepted message (see \cref{s:3.6}). We say that the broadcast (to be understood as a pair of broadcast phases) is \emph{faulty} if the message of highest priority fails to reach the leader in the VHT broadcast phase or fails to be delivered to all processes in the acknowledgment broadcast phase.

\myskip
\mypar{Ideal VHT and effective VHT.} Suppose that, at some point during the execution of our algorithm, the $t$ rounds $N_1$, $N_2$, \dots, $N_t$ of the virtual network $\mathcal N$ have been determined from the selected networks $G_{i_1}$, $G_{i_2}$, \dots, $G_{i_t}$ as explained in \cref{s:3.4}. These $t$ rounds of the virtual network define the levels of the virtual history tree from level $L_{-1}$ up to level $L_t$. We refer to the tree composed of these levels as the \emph{ideal VHT}, because it is what the algorithm aims to construct.

Since some broadcasts may be faulty, each new level added to the variable \texttt{VHT} may be missing some parts compared to the ideal VHT. We refer to the actual content of the variable \texttt{VHT} as the \emph{effective VHT}. Note that, while the nodes of the effective VHT have IDs, the nodes of the ideal VHT do not, except in level $L_0$ (see \cref{fig:ht1}). Still, the effective VHT may be \emph{isomorphic} to a subgraph of the ideal VHT when we remove IDs from nodes, except for the IDs~$0$ and~$1$ (which are in level $L_0$).

\myskip
\mypar{Generalized views.} A \emph{generalized view} $\mathcal V$ of a history tree $\mathcal H$ is a subgraph of $\mathcal H$ such that, if $\mathcal V$ contains a node $v$ of $\mathcal H$, then $\mathcal V$ also contains the nodes and edges spanned by all shortest paths in $\mathcal H$ (using black and red edges indifferently) from the root of $\mathcal H$ to $v$. As the name implies, a view is a special case of a generalized view.

\myskip
\mypar{Construction of the VHT.} We will now prove that our VHT construction protocol is correct.

\begin{lemma}\label{l:vhtcorr}
At all times, the effective VHT is isomorphic to a generalized view of the ideal VHT.
\end{lemma}
\begin{proof}
The proof is by induction on the number of nodes added to the effective VHT via the function \texttt{UpdateVHT} in \al{2}{27}. By the inductive hypothesis, there is an isomorphism $\varphi$ that maps nodes of the effective VHT to equivalent nodes of a generalized view of the ideal VHT (disregarding IDs, except in level $L_0$).

As part of the inductive proof, we also include an additional statement: When level $L_t$ of the effective VHT is finalized, let $p$ be any process, and let $v'_p$ be the node of level $L_t$ of the ideal VHT representing $p$. If there is a node $v_p$ of the effective VHT such that $\varphi(v_p)=v'_p$, then $p$ and $v_p$ have the same ID. Otherwise, $p$ is an error process. Clearly, all of these claims hold for level $L_0$.

We will now show how to extend $\varphi$ to a new node $u$ created by a call to \texttt{UpdateVHT}, say in level $L_{t+1}$ of the effective VHT. Let $q$ be any process that sent a Done message with the ID passed to \texttt{UpdateVHT}, and let $u'$ be the node representing $q$ in level $L_{t+1}$ the ideal VHT.

By the inductive hypothesis, the IDs and multiplicities stored in the \texttt{ObsList} of $q$ during the begin round for level $L_{t+1}$, say round $i_{t+1}$, correctly encoded red edges incident to $u'$. There were possibly some extra elements in \texttt{ObsList}, but these were removed after some broadcast phases by the function \texttt{PreventCyclesInLevelGraph}. Hence, any red-edge triplet that was broadcast by $q$ had the correct meaning.

Note that none of the neighbors of $q$ in the virtual network $N_{t+1}$ with an ID different from the ID of $q$ was an error process at round $i_{t+1}$, or else $q$ would have become an error process as well, and would not have participated in any VHT broadcast. So, $q$ had \emph{all} the red-edge triplets corresponding to its neighbors with different IDs. $q$ also had the red-edge triplet corresponding to the cycle $C_w$ involving it ($w$ being the node of level $L_t$ of the effective VHT representing $q$), because the pair $(\texttt{MyID}, 2)$ is automatically inserted in \texttt{ObsList}.

We conclude that, if $q$ sent a Done message, it must have successfully sent the red-edge triplets corresponding to all the red edges incident to $u'$. All of these were accepted by the leader. Thus, all the corresponding red edges we added to the temporary VHT, and are now being added to the effective VHT, incident to the new node $u$. Also, $u$ must be a child of $w$, the node that represented $q$ in $L_t$. Therefore, $\varphi$ can be extended as $\varphi(u)=u'$, because $u$ and $u'$ have equivalent incident black and red edges.

By the inductive hypothesis, the effective VHT was isomorphic to a generalized view of the ideal VHT before the addition of $u$. Since the edges incident to $u$ in the effective VHT are equivalent to \emph{all} the edges incident to $u'$ in the ideal VHT, the updated effective VHT is still isomorphic to a generalized view.

It remains to prove that, when level $L_{t+1}$ of the effective VHT is finalized, any process $q'$ that is not represented by any node of the effective VHT must be an error process. That is, no node in level $L_{t+1}$ of the effective VHT has the same ID as $q'$. This implies that $q'$ never sent a Done message, or its Done message was never accepted by the leader. In any case, $q'$ must be an error process.
\end{proof}

\begin{corollary}\label{c:vhtcorr}
Extracting the leader's view from the effective VHT (\al{2}{30}) yields a graph isomorphic to the view of the leader in the ideal VHT.
\end{corollary}
\begin{proof}
By \cref{l:vhtcorr}, the effective VHT is always isomorphic to a generalized view of the ideal VHT. The instruction in \al{2}{30} is executed by the leader after completing a new level of the effective VHT. Since the leader did not enter a reset phase, it must be represented in the effective VHT; that is, there is a node that represents the leader in the new level. By definition of generalized view, since the effective VHT contains the leader's node, it also contains all the nodes and edges corresponding to the view of the leader in the ideal VHT. Thus, extracting the leader's view from the effective VHT yields a graph isomorphic to the view of the leader in the ideal VHT.
\end{proof}

\mypar{Size of the VHT.} Recall that links are not included in the virtual network if they cause cycles to appear in \texttt{LevelGraph} (see \cref{s:3.4}). Although this strategy may not guarantee that every level of the VHT has $O(n)$ red edges, it does amortize the total number of red edges over several levels.

\begin{lemma}\label{l:vhtsize}
The total number of red edges in the first $O(n)$ levels of the ideal VHT, counted \emph{without their multiplicity}, is $O(n^2)$.
\end{lemma}
\begin{proof}
By definition of history tree, each node of the VHT represents at least one process, and the nodes in each level represent a partition of the $n$ processes. It follows that there can be at most $n$ nodes in each level of the VHT. Thus, in $O(n)$ levels, at most $O(n^2)$ red edges can connect a node to its parent: These are the edges corresponding to the cycles $C_v$, which are drawn in green in \cref{fig:ht2}. In the following, we will count the remaining red edges.

For any $t\geq 1$, let $L_t$ be the level of the VHT of index $t$. For $1\leq i\leq |L_t|$, let $d_{t,i}$ be the degree of the $i$th node of \texttt{LevelGraph} when the construction of level $L_{t+1}$ is finished. Since \texttt{LevelGraph} is a tree on $|L_t|\leq n$ nodes, we have $\sum_{i=1}^{|L_t|} d_{t,i}\leq 2n$.

Let $v_{t,i}$ be the $i$th node in $L_t$, and let $c_{t,i}$ be the number of children of $v_{t,i}$. Note that there can be at most $n$ branches in the VHT, because each level has at most $n$ nodes, and the number of nodes per level cannot decrease. Since $c_{t,i}-1$ is the number of \emph{new} branches stemming from $v_{t,i}$, we have $\sum_{t=1}^{\infty}\sum_{i=1}^{|L_t|} (c_{t,i} - 1) \leq n$.

Let $v_{t,i},v_{t,j}\in L_t$. Observe that there is an edge in \texttt{LevelGraph} connecting these two nodes if and only if some children of $v_{t,i}$ are connected to $v_{t,j}$ by red edges of the VHT; the number of these red edges is at most $c_{t,i}$. Therefore, the red edges connecting $L_t$ and $L_{t+1}$ that are incident to the children of $v_{t,i}$ are at most $c_{t,i}\cdot d_{t,i}$. Now, let $R_m$ be the total number of red edges in the first $m$ levels of the VHT. Defining $c^\ast_{t}=\max_{i=1}^{|L_t|} c_{t,i}$, we have $\sum_{t=1}^{\infty} (c^\ast_t-1)\leq n$, and thus
$$R_m\leq \sum_{t=1}^{m}\sum_{i=1}^{|L_t|} c_{t,i}d_{t,i} \leq \sum_{t=1}^{m}c^\ast_t\sum_{i=1}^{|L_t|} d_{t,i}\leq 2n \sum_{t=1}^{m} c^\ast_t \leq 2n\left(m + \sum_{t=1}^{m} (c^\ast_t-1) \right)\leq 2n(m+n).$$
Plugging $m=O(n)$, we obtain our claim.
\end{proof}

By comparison, the first $O(n)$ levels of a generic history tree may contain $\Theta(n^3)$ red edges: Think of a network that makes all $n$ processes distinguishable at the first round and has the topology of the complete graph $K_n$ in all subsequent rounds.

\myskip
\mypar{Main theorem.} We will now prove that our algorithm is correct and runs in $O(n^3\log n)$ rounds.

\begin{lemma}\label{l:maxdiam}
The value of the local variable \texttt{DiamEstimate} never exceeds $4n$.
\end{lemma}
\begin{proof}
It is well known that the dynamic diameter of a connected network of size $n$ is smaller than $n$ (see~\cite{KLO10}), and \texttt{DiamEstimate} is doubled after every reset phase. Thus, as soon as $n\leq \texttt{DiamEstimate}\leq 2n$, no broadcast can be faulty, assuming there are no error processes left in the network. In case there are still Error messages circulating, the leader receives the one with highest priority within the next broadcast phase. Thus, after one additional reset phase, which involves all processes, $\texttt{DiamEstimate}\leq 4n$ and no further resets can occur.
\end{proof}

\begin{theorem}\label{t:main}
The Counting algorithm of \cref{s:3} is correct and runs in $O(n^3 \log n)$ rounds.
\end{theorem}
\begin{proof}
By construction, the virtual network $N_t$ is connected for every $t\geq 1$. Indeed, with the notation used in \cref{s:3.4}, the minor of $N_t$ obtained by contracting the edges of the cycles $C_v$ is isomorphic to the tree $S$.

Also, by \cref{c:vhtcorr}, the function \texttt{CountFromView} is always called on a graph isomorphic to the view of the leader in the virtual network. As proved in~\cite{DV22}, when this function is given as input a view of the history tree of a connected network, it either returns ``Unknown'' or the number of processes in the network, $n$. Moreover, if the input view spans at least $3n$ levels, the returned value is necessarily $n$. This also implies that the algorithm cannot terminate early with an incorrect result.

Hence, we only have to prove that the first $3n$ levels of the effective VHT are constructed within the desired number of rounds.

By \cref{l:vhtsize}, the first $3n$ levels of the ideal VHT contain $O(n^2)$ red edges. This is also true of the effective VHT, because it is isomorphic to a subtree of the ideal VHT by \cref{l:vhtcorr}. Thus, it takes $O(n^2)$ broadcasts to construct the red edges in the first $3n$ levels of the effective VHT (via \texttt{UpdateTempVHT}). It also takes $O(n^2)$ broadcasts to construct all the nodes in the first $3n$ levels of the effective VHT (via \texttt{UpdateVHT}). In turn, each broadcast takes $O(n)$ rounds, because its duration is proportional to \texttt{DiamEstimate}, which is at most $4n$ by \cref{l:maxdiam}. We conclude that, every $O(n^3)$ rounds, either the first $3n$ levels of the effective VHT are constructed or there is a reset that deletes some levels.

Since \texttt{DiamEstimate} doubles at every reset and never exceeds $4n$, there can be at most $O(\log n)$ resets in total. Also, each reset phase has a duration proportional to \texttt{DiamEstimate}, hence $O(n)$ rounds. Therefore, the first $3n$ levels of the effective VHT are constructed within $O(n^3\log n)$ rounds. At this point, the leader terminates and gives the correct output, $n$.
\end{proof}

Since our Counting algorithm has a polynomial running time, any ``timestamp'' can be encoded in $O(\log n)$ bits. It is now straightforward to conclude that the algorithm works in the congested network model.

\begin{corollary}\label{c:msize}
The Counting algorithm of \cref{s:3} works in the congested network model.
\end{corollary}
\begin{proof}
Recall from \cref{s:3.2} that each message contains a constant-sized label and at most three parameters. Some parameters may be IDs of processes, indices of levels of the VHT that are or have been under construction, or round numbers. For all these parameters, the maximum possible value increases by at most one unit per round, and therefore it is at most $O(n^3\log n)$ by \cref{t:main}. Some parameters may denote multiplicities of red edges, which are bounded by a polynomial because such is our assumption about the total number of communication links in a single round (see \cref{s:2}). Finally, a parameter may denote an estimate on the dynamic diameter, which is always $O(n)$ by \cref{l:maxdiam}. Since all parameters have polynomial size, they can be included in messages of $O(\log n)$ bits, and are therefore suitable for the congested network model.
\end{proof}

\section{Extensions and Improvements}\label{s:5}

\mypar{Simultaneous termination.} In \cref{s:3} we showed how the leader can determine $n$. If we want all processes, not just the leader, to simultaneously output $n$ and terminate, we can use the following protocol. As soon as the leader knows $n$, it broadcasts a message of maximum priority containing $n$ and the current round number $c$. Any process that receives this message keeps forwarding it until round $c+n$. By this round they all have it, and thus can output $n$ and terminate simultaneously.

\myskip
\mypar{General computation.} If each process is assigned an $O(\log n)$-sized input (other than the leader flag), we can also perform general computations on the multiset of inputs. We only have to adapt the algorithm of \cref{s:3} to work with history trees whose level $L_0$ contains one node per input occurring in the system, and then run the algorithm in~\cite{DV22} on this extended history tree. This can easily be done with the techniques developed in \cref{s:3}: We can construct $L_0$ like any other level, except that we use broadcast phases to transmit inputs instead of red edges.

\myskip
\mypar{Optimized running time.} We can sightly speed up our algorithm at the cost of some additional bookkeeping. The idea is to refine the error and reset mechanism by resuming the construction of the VHT not from the \emph{level} that caused an error, but from the \emph{broadcast phase} that caused it.

The processes have to remember the order in which the leader accepted Red-edge and Done messages, as well as the Begin messages received at every begin round. When a process detects an error, it broadcasts an Error message containing not the current level number, but the number of messages that the leader has accepted up to that time. The advantage is that the reset phase can rewind the (temporary) VHT exactly to the desired point without erasing entire levels. Furthermore, by looking up the Begin messages received in the appropriate begin round, each process is also able to reconstruct its local \texttt{ObsList} at the desired time.

If a broadcast phase causes an error, the leader receives a notification after less than $n$ rounds, and therefore it only has to undo the work done in $O(n)$ rounds. Since there can be at most $O(\log n)$ resets, the total time spent in reset phases or doing work that later gets undone is at most $O(n\log n)$ rounds. Thus, the new algorithm runs in $O(n^3) +O(n\log n)=O(n^3)$ rounds.

\myskip
\mypar{Leaderless computation.} It is known that doing any non-trivial computation in a leaderless network requires some extra assumptions, for example the knowledge of an upper bound $D$ on its dynamic diameter~\cite{DVdisc}. However, knowing $D$ allows processes to reliably broadcast messages in phases of $D$ rounds. This immediately yields an extension of our algorithm to leaderless networks in $O(Dn^2)$ rounds, where no acknowledgment phases or error and reset phases are needed.

\myskip
\mypar{Disconnected networks.} We can extend our algorithm to $T$-union-connected networks, assuming the parameter $T$ is known. As discussed in~\cite{DVdisc}, the idea is to divide the sequence of rounds into blocks of size $T$. Within every block, each process keeps sending the same message and stores all incoming messages. At the end of a block, each process runs the algorithm from \cref{s:3} (or its optimization) pretending that all the stored messages arrived in a single round. This is equivalent to running the algorithm on the dynamic network $\mathcal G^\star=(G^\star_{1}, G^\star_{T+1}, G^\star_{2T+1},\dots)$, which is always connected ($G^\star_t$ is defined in \cref{s:2}). The running time is simply $O(Tn^3)$ rounds.

\section{Concluding Remarks}\label{s:6}
In this paper we have extended the theory of history trees by introducing the tools necessary for the distributed construction and transmission of history trees in the congested network model. This resulted in a new state of the art for general computation in disconnected anonymous dynamic congested networks, with or without leaders.

Our history tree construction technique leads to general algorithms whose running time is cubic in the size of the network. An immediate open problem is whether this running time can be reduced. Since our algorithm broadcasts information using a token-forwarding approach, and by virtue of the $\Omega(n^2/\log n)$ lower bound of~\cite{DPR13}, we believe that it would be unlikely to achieve a better running time without a radical change in the technique used.

Understanding whether the Counting problem has as a super-linear lower bound in congested networks is of special importance, because it would mark a computational difference between anonymous dynamic networks in the congested and non-congested models.

It would be interesting to do a thorough fine-grained tradeoff analysis of our algorithm. For instance, it is not difficult to show that, if messages have size $O(n \log n)$, the running time of our algorithm can be reduced to $O(n^2)$.

\bibliographystyle{ACM-Reference-Format}
\bibliography{congested}

\newpage
\appendix

\section{Pseudocode of the Basic Counting Algorithm}\label{a:1}

This section contains the full pseudocode for the basic Counting algorithm described in \cref{s:3}. The entry point is the function \texttt{Main} in \cref{l:2}. Each process runs an independent instance of this function and has private instances of the variables in \cref{l:1}.

The only primitives are the following functions:
\begin{itemize}
\item \texttt{Input} returns the process' input, i.e., whether the process is the leader.
\item \texttt{Output} is used to give an output. It should be called only once at the end of the execution.
\item \texttt{SendToAllNeighbors} takes a message as an argument and sends it to all processes that share a link with the caller in the current communication round.
\item \texttt{ReceiveFromAllNeighbors} returns a multiset of messages coming from all incident links. These are the messages that have been passed to the function \texttt{SendToAllNeighbors} by the neighboring processes in the current round.
\item \texttt{CountFromView} takes a view of a history tree and runs the Counting algorithm from~\cite{DV22} on it. The returned value is either ``Unknown'' or the total number of processes in a network represented by that history tree.
\end{itemize}

\lstset{style=mystyle}

\begin{lstlisting}[caption={Internal variables and their initialization\label{l:1}},captionpos=t,float=ht,mathescape=true]
# Internal variables of each process
LeaderFlag       # indicates whether the process is the leader (Boolean)
MyID             # temporary ID (non-negative integer)
NextFreshID      # next ID to be assigned (positive integer)
CurrentRound     # current communication round (non-negative integer)
VHT              # first levels of the virtual history tree
CurrentLevel     # level of the VHT being constructed (positive integer)
TempVHT          # level under construction
ObsList          # list of observations to be broadcast
LevelGraph       # auxiliary graph used in the construction of the virtual network
DiamEstimate     # guess on the dynamic diameter (integer power of 2)
NetworkSize      # exact number of processes (integer or "Unknown")

function InitializeVariables()
    LeaderFlag := Input()
    if LeaderFlag = True then MyID := 0
    else MyID := 1
    NextFreshID := 2
    CurrentRound := 0
    Root := new node
    Root.ID := -1
    Node0 := new node
    Node0.ID := 0
    Node0.Parent := Root
    Node1 := new node
    Node1.ID := 1
    Node1.Parent := Root
    VHT := history tree consisting of Root, Node0, Node1
    CurrentLevel := 1
    DiamEstimate := 1
    NetworkSize := "Unknown"
\end{lstlisting}

\begin{lstlisting}[caption={Algorithm entry point and main loop\label{l:2}},captionpos=t,float=ht,mathescape=true]
# Input: whether or not the process is the leader
# Output: the number of processes in the system, $n$

function Main()
    InitializeVariables()
    do
        do
            ReturnValue := SetUpNewLevel()
        while ReturnValue = "Restart"
        do
            OriginalMessage := MakeVHTMessage()
            VHTMessage := BroadcastPhase(OriginalMessage)
            if VHTMessage = "Restart" then goto Line 7
            if LeaderFlag then
                AckMessage := BroadcastPhase(VHTMessage)
            else
                NullMessage := new message
                NullMessage.Label := "Null"
                AckMessage := BroadcastPhase(NullMessage)
            if AckMessage = "Restart" then goto Line 7
            if AckMessage != VHTMessage then
                BroadcastError(CurrentLevel)
                goto Line 7
            if AckMessage.Label = "Edge" then
                UpdateTempVHT(AckMessage.ID1, AckMessage.ID2, AckMessage.Mult)
            if AckMessage.Label = "Done" then
                UpdateVHT(AckMessage.ID)
        while AckMessage.Label != "End"
        if LeaderFlag = True then
            View := view of the leader in VHT
            NetworkSize := CountFromView(View)     # Counting algorithm from $\text{\cite{DV22}}$
        CurrentLevel := CurrentLevel + 1
    while NetworkSize = "Unknown"
    Output(NetworkSize)
\end{lstlisting}

\begin{lstlisting}[caption={Communication and broadcast functions\label{l:3}},captionpos=t,float=ht,mathescape=true]
function EdgePriority(ID1, ID2, Mult)
    return 1 / (2^ID1 * 3^ID2 * 5^Mult)

function Priority(Message)
    if Message.Label = "Null" then return 0
    if Message.Label = "Begin" then return 1
    if Message.Label = "End" then return 2
    if Message.Label = "Done" then return 2 + (1 / Message.ID)
    if Message.Label = "Edge" then
        return 3 + EdgePriority(Message.ID1, Message.ID2, Message.Mult)
    if Message.Label = "Error" then return 4 + (1 / (2 * Message.ErrorLevel + 1))
    if Message.Label = "Reset" then return 4 + (1 / (2 * Message.ResetLevel))

function SendAndReceive(Message)
    SendToAllNeighbors(Message)
    MessageMultiset := ReceiveFromAllNeighbors()
    CurrentRound := CurrentRound + 1
    return MessageMultiset

function BroadcastStep(Message)
    TopMessage := Message
    MessageMultiset := SendAndReceive(Message)
    forall ReceivedMessage in MessageMultiset do
        if Priority(ReceivedMessage) > Priority(TopMessage) then
            TopMessage := ReceivedMessage
    return TopMessage

function BroadcastPhase(Message)
    TopMessage := Message
    for i := 1 to DiamEstimate do
        TopMessage := BroadcastStep(TopMessage)
    if TopMessage.Label = "Error" then
        HandleError(TopMessage)
        return "Restart"
    else if TopMessage.Label = "Reset"
        BroadcastReset(TopMessage)
        return "Restart"
    else return TopMessage
\end{lstlisting}

\begin{lstlisting}[caption={Creation and selection of the observations to be broadcast\label{l:4}},captionpos=t,float=ht,mathescape=true]
function SetUpNewLevel()
    BeginMessage := new message
    BeginMessage.Label := "Begin"
    BeginMessage.ID := MyID
    MessageMultiset := SendAndReceive(BeginMessage)
    ObsList := empty list
    forall (Message $\times$ Mult) in MessageMultiset do
        if Message.Label != "Begin" then
            HandleError(Message)
            return "Restart"
        if MyID != Message.ID then
            Add (Message.ID, Mult) to ObsList
    Add (MyID, 2) to ObsList
    TempVHT := empty graph
    LevelGraph := empty graph
    forall Node in level (CurrentLevel - 1) of VHT do
        Add a copy of Node to TempVHT
        Add a copy of Node to LevelGraph
    return "OK"

function MakeVHTMessage()
    Message := new message
    if (ObsList is empty) then
        if (VHT contains a node whose ID is MyID) then
            Message.Label := "End"
        else
            Message.Label := "Done"
            Message.ID := MyID
    else
        Message.Label := "Edge"
        (ID2, Mult) := first element in ObsList
        Message.ID1 := MyID
        Message.ID2 := ID2
        Message.Mult := Mult
    return Message
\end{lstlisting}

\begin{lstlisting}[caption={Construction of the virtual history tree (VHT)\label{l:5}},captionpos=t,float=ht,mathescape=true]
function FindRoot(ID)
    Node := node in TempVHT such that Node.ID = ID
    Tree := tree of TempVHT containing Node
    Root := root of TempTree
    return Root

function PreventCyclesInLevelGraph()
    forall (ID2, Mult) in ObsList do
        Root := FindRoot(MyID)
        Node1 := node of LevelGraph with the same ID as Root
        Node2 := node in LevelGraph such that Node2.ID = ID2
        if (Node1 != Node2) and ({Node1, Node2} is not an edge of LevelGraph) then
            LevelGraph' := copy of LevelGraph
            Add edge {Node1, Node2} to LevelGraph'
            if (LevelGraph' has a cycle) then Remove (ID2, Mult) from ObsList

function UpdateTempVHT(ID1, ID2, Mult)
    Node := node of TempVHT such that Node.ID = ID1
    Root1 := FindRoot(ID1)
    Root2 := FindRoot(ID2)
    Child := new node
    Child.ID := NextFreshID
    NextFreshID := NextFreshID + 1
    Child.Parent := Node
    Add red edge {Child, Root2} $\times$ Mult to TempVHT
    if (MyID = ID1) and ((ID2, Mult) is in ObsList) then
        Remove (ID2, Mult) from ObsList
        MyID := Child.ID
    LGNode1 := node of LevelGraph with the same ID as Root1
    LGNode2 := node of LevelGraph with the same ID as Root2
    if (LGNode1 != LGNode2) and ({LGNode1, LGNode2} is not an edge of LevelGraph) then
        Add edge {LGNode1, LGNode2} to LevelGraph
        PreventCyclesInLevelGraph()

function UpdateVHT(ID)
    TempNode := node of TempVHT such that TempNode.ID = ID
    TempRoot := FindRoot(TempNode.ID)
    VHTParent := node of VHT with the same ID as TempRoot
    VHTChild := new node
    VHTChild.ID := TempNode.ID
    VHTChild.Parent := VHTParent
    IterNode := TempNode
    while IterNode != TempRoot do
        forall {Node1, Node2} $\times$ Mult in {red edges of TempVHT} do
            if Node1.ID = IterNode.ID then
                VHTObsNode := node of VHT with the same ID as Node2
                Add red edge {VHTChild, VHTObsNode} $\times$ Mult to VHT
        IterNode := IterNode.Parent
\end{lstlisting}

\begin{lstlisting}[caption={Reacting to errors and performing resets\label{l:6}},captionpos=t,float=ht,mathescape=true]
function MakeResetMessage()
    ResetMessage := new message
    ResetMessage.Label := "Reset"
    ResetMessage.ResetLevel := CurrentLevel
    ResetMessage.StartingRound := CurrentRound
    ResetMessage.NewDiam := DiamEstimate * 2
    return ResetMessage

function HandleError(Message)
    if (Message.Label = "Error") and (Message.ErrorLevel < CurrentLevel) then
        CurrentLevel := Message.ErrorLevel
    if LeaderFlag = True then
        NullMessage := new message
        NullMessage.Label := "Null"
        for i := 0 to DiamEstimate * 2 do
            SendAndReceive(NullMessage)
        ResetMessage := MakeResetMessage()
        BroadcastReset(ResetMessage)
    else BroadcastError(CurrentLevel)

function BroadcastError(ErrorLevel)
    Message := new message
    Message.Label := "Error"
    Message.ErrorLevel := ErrorLevel
    while Message.Label != "Reset" do
        Message := BroadcastStep(Message)
    BroadcastReset(Message)

function BroadcastReset(ResetMessage)
    FinalRound := ResetMessage.StartingRound + ResetMessage.NewDiam
    TopMessage := ResetMessage
    while CurrentRound < FinalRound do
        TopMessage := BroadcastStep(TopMessage)
    TempRoot := FindRoot(MyID)
    VHTNode1 := node of VHT with the same ID as TempRoot
    VHTNode2 := ancestor of VHTNode1 in level (ResetMessage.ResetLevel - 1)
    MyID := VHTNode2.ID
    Delete all nodes in VHT whose level is at least ResetMessage.ResetLevel
    (Also delete all edges incident to deleted nodes)
    CurrentLevel := ResetMessage.ResetLevel
    DiamEstimate := ResetMessage.NewDiam
\end{lstlisting}

\end{document}